\documentclass[11pt]{article}

\usepackage{amssymb} % set of additional math symbols
\usepackage{amsmath}     % new environments
\usepackage{amsthm}
\usepackage{todonotes}

\usepackage[ruled,vlined,linesnumbered]{algorithm2e}
\SetKw{KwNot}{not}
\SetKw{KwOR}{or}
\SetKw{KwExitFor}{exit for-loop}

\usepackage{a4wide}
\usepackage{graphicx}
\usepackage{hyperref}

\newtheorem{thm}{Theorem}
\newtheorem{lem}{Lemma}

\newtheorem{obs}{Observation}

\newtheorem{cla}{Claim}

\newcommand{\Rset}{\mathbb{R}}
\newcommand{\Zset}{\mathbb{Z}}
\newenvironment{keyword}{\par{\noindent\bf Keywords:}}

\begin{document}

%*******TITLE AND AUTHORS*******************************************
\title{The robust recoverable spanning tree problem with interval costs is polynomially solvable} 

\author{Mikita Hradovich$^\dag$, Adam Kasperski$^\ddag$,  Pawe{\l} Zieli{\'n}ski$^\dag$\\
          {\small \textit{$^\dag$Faculty of Fundamental Problems of Technology,}}\\
	{\small \textit{Wroc{\l}aw University of Technology,  Wroc{\l}aw, Poland}}\\
	 {\small \textit{$^\ddag$Faculty of Computer Science and Management,}}\\
	{\small \textit{Wroc{\l}aw University of Technology,  Wroc{\l}aw, Poland}}\\
	{\small \texttt{\{mikita.hradovich,adam.kasperski,pawel.zielinski\}@pwr.edu.pl}}
}

  \date{}
    
\maketitle

\begin{abstract}
In this paper the robust recoverable spanning tree problem with interval edge costs is considered. The complexity of this problem has remained open to date. It is shown that the problem is polynomially solvable, by using an iterative relaxation method. A generalization of this idea to the robust recoverable matroid basis problem is also presented.
Polynomial algorithms for both  robust recoverable problems are proposed.
\end{abstract}


\section{Introduction}

In this paper, we wish to investigate the robust recoverable version of the following \emph{minimum spanning tree} problem.
We are given a connected graph~$G=(V,E)$, where $|V|=n$ and $|E|=m$. Let $\Phi$ be the set of all \emph{spanning trees} of~$G$. For each edge $e\in E$ a nonnegative cost $c_e$ is given. We seek a spanning tree of~$G$  of the minimum total cost. The minimum spanning tree problem can be solved in polynomial time by several well known algorithms~(see, e.g.~\cite{AMO93}). 
In this paper we consider the \emph{robust recoverable model}, previously discussed in~\cite{B11,B12,LLMS09}. We are given first stage edge costs $C_e$, $e\in E$, recovery parameter $k\in\{0,\dots,n-1\}$, and uncertain second stage (recovery) edge costs, modeled by scenarios.  Namely, each particular realization of the second stage costs $S=(c_e^S)_{e\in E}$ is called a \emph{scenario} and the set of all possible scenarios is denoted by $\mathcal{U}$.
In the \emph{robust recoverable spanning tree} problem (RR~ST, for short), we choose an initial spanning tree $X$ in the first stage. The cost of this tree is equal to $\sum_{e\in X} C_e$. Then, after scenario $S\in\mathcal{U}$ reveals,  $X$ can be modified by exchanging up to $k$ edges. This new tree is denoted by $Y$, where $|Y\setminus X|=|X\setminus Y| \leq k$. The second stage cost of $Y$ under scenario $S\in \mathcal{U}$ is equal to $\sum_{e\in Y} c_e^S$. Our goal is to find a pair of trees $X$ and $Y$ such that $|X\setminus Y|\leq k$, which minimize the total fist and second stage cost $\sum_{e\in X} C_e +\sum_{e\in Y} c_e^S$ in the worst case. Hence, the problem RR~ST can be formally stated as follows:

\begin{equation}
\textsc{RR ST}: \; \min_{X\in \Phi} \left(\sum_{e\in X} C_e + \max_{S\in \mathcal{U}}\min_{Y\in \Phi^k_{X}} \sum_{e\in Y} c_e^S\right),
\label{rrp}
\end{equation}
where 
$\Phi^k_{X}=\{Y\in\Phi \,:\,  |Y\setminus X|\leq k\}$ is the \emph{recovery set}, i.e. the set of possible solutions in 
the second, recovery stage.

The RR~ST problem has been recently discussed in a number of papers. It is a special case of the robust spanning tree problem with  incremental recourse considered in~\cite{NO13}. Furthermore, if $k=0$ and $C_e=0$ for each $e\in E$, then the problem is equivalent to the robust min-max spanning tree problem investigated in~\cite{KY97, KZ11}.
 The complexity of RR~ST depends on the way in which scenario set $\mathcal{U}$ is defined. If $\mathcal{U}=\{S_1,\dots,S_K\}$ contains $K\geq 1$, explicitly listed scenarios, then the problem is known to be NP-hard for $K=2$ and any constant $k\in\{0,\dots,n-1\}$~\cite{KKZ14}. Furthermore, it becomes strongly NP-hard and not at all approximable when both $K$ and $k$ are  part of input~\cite{KKZ14}. Assume now that the second stage cost of each edge $e\in E$ is known to belong to the closed interval $[c_e, c_e+d_e]$, where $d_e\geq 0$. Scenario set $\mathcal{U}^l$ is then the subset of the Cartesian product $\prod_{e\in E}[c_e, c_e+d_e]$ such that in each scenario in $\mathcal{U}^l$, the costs of at most $l$ edges are greater than their nominal values $c_e$,
$l\in\{0,\dots,m\}$.
 Scenario set $\mathcal{U}^l$ has been proposed in~\cite{BS03}. The parameter $l$ allows us to model the degree of uncertainty. Namely, if $l=0$ then $\mathcal{U}$ contains only one scenario. 
 The problem RR~ST for scenario set $\mathcal{U}^l$ is known to be strongly NP-hard when $l$ is a part of input~\cite{NO13}. In fact, the inner problem, $\max_{S\in \mathcal{U}^l}\min_{Y\in \Phi^k_X} \sum_{e\in Y} {c_e^S}$, called the \emph{adversarial problem}, is then strongly NP-hard~\cite{NO13}. 
 On the other hand, $\mathcal{U}^m$ is the Cartesian product of all the uncertainty intervals, and represents the traditional interval uncertainty representation~\cite{KY97}. 

The complexity of RR~ST with scenario set $\mathcal{U}^m$ is open to date. 
In~\cite{SAO09} the \emph{incremental spanning tree} problem was discussed. In this problem we are given an initial spanning tree $X$ and we seek a spanning tree $Y\in \Phi^k_X$ whose total cost is minimal. It is easy to see that this problem is the inner one in RR~ST, where $X$ is fixed and $\mathcal{U}$ contains only one scenario. The incremental spanning tree problem can be solved in polynomial time by applying the Lagrangian relaxation technique~\cite{SAO09}. 
In~\cite{B11} a polynomial algorithm for  a more general \emph{recoverable matroid basis} problem
(RR~MB, for short)
 with scenario set $\mathcal{U}^m$
 was proposed, provided that the recovery parameter~$k$ is constant
and, in consequence, for   RR~ST (a spanning tree is a  graphic matroid).
Unfortunately,  the algorithm is exponential in~$k$. 
No other result on the problem is known to date. In particular, no polynomial time algorithm  has been developed
when~$k$ is a part of the input.

In this paper we show that RR~ST for the interval uncertainty representation (i.e. for scenario set $\mathcal{U}^m$) is polynomially solvable (Section~\ref{srrsp}). We apply a technique called the \emph{iterative relaxation}, whose framework was described in~\cite{LRS11}. The idea is to construct a linear programming relaxation of the problem and show that at least one variable in each optimum vertex solution is integer. Such a variable allows us to add an edge to the solution built and recursively solve the relaxation of the smaller problem. We also show that this technique allows us to solve the recoverable matroid basis problem (RR~MB) for the interval uncertainty representation in polynomial time
(Section~\ref{srrmb}).
We provide polynomial algorithms for RR~ST and RR~MB.

\section{Robust recoverable  spanning tree problem}
\label{srrsp}

In this section we will  use 
the iterative relaxation method~\cite{LRS11} to construct a polynomial algorithm 
for RR~ST under scenario set $\mathcal{U}^m$. Notice first that, in this case, the formulation~(\ref{rrp}) can be rewritten as follows:
\begin{align}
\textsc{RR ST}: \; \min_{X\in\Phi}
\left (\sum_{e\in X} C_e  +    \min_{Y\in \Phi^{k}_{X}}\sum_{e\in Y}(c_e+d_e) \right).
\label{rrst}
\end{align}
In problem~(\ref{rrst}) we need to find a pair of spanning tree $X\in \Phi$ and $Y\in\Phi^k_X$. 
Since $|X|=|Y|=|V|-1$, the problem~(\ref{rrst}) is equivalent the following mathematical programming problem: 
\begin{equation}
 \begin{array}{ll}
 \min &\sum_{e\in X}C_e + \sum_{e\in Y}(c_e+d_e)\\
\text{s.t. } & |X \cap Y| \geq |V|-1-k, \\
       & X,Y\in \Phi.
 \end{array}
 \label{mpst}
\end{equation}

We now set up some additional notation. Let $V_X$ and $V_Y$ be subsets of vertices $V$, and $E_X$ and $E_Y$ be subsets of edges $E$, which induce connected graphs (multigraphs) 
 $G_X=(V_X, E_X)$ and $G_Y=(V_Y,E_Y)$, respectively. Let $E_Z$ be a subset of $E$ such that $E_Z\subseteq E_X \cup E_Y$ and $|E_Z|\geq L$ for some fixed
integer~$L$. We will use $E_X(U)$ (resp. $E_Y(U)$) to denote the set of edges that has both endpoints in a given subset of vertices $U\subseteq V_X$ (resp. $U\subseteq V_Y$).

Let us consider the following linear program, denoted by $LP_{RRST}(E_{X},V_{X},E_{Y},V_{Y}, E_{Z},L)$,
that we will substantially use in the algorithm for RR~ST:
\begin{eqnarray}
    \min        \sum_{e\in E_X}C_e x_e + \sum_{e\in E_{Y}}(c_e+d_e) y_e&& \label{lpst_o}\\
    \text{s.t. }\;\;  \sum_{e\in E_X} x_e = |V_X|-1, && \label{lpst_0} \\
                      \sum_{e\in E_X(U)} x_e \leq  |U|-1, && \forall U\subset V_X, \label{lpst_1}\\
                      -x_e+z_e \leq  0,                              &&\forall e\in E_{X}\cap E_{Z},  \label{lpst_2}\\
                      \sum_{e\in E_Z} z_e= L, &&  \label{lpst_3}\\
                       z_e -y_e \leq  0,                              &&\forall e\in E_{Y}\cap E_{Z},  \label{lpst_4}\\
                       \sum_{e\in E_Y} y_e = |V_Y|-1, &&   \label{lpst_5}\\
                      \sum_{e\in E_Y(U)} y_e \leq  |U|-1, && \forall U\subset V_Y,  \label{lpst_6}\\
                      x_e\geq 0, && \forall e\in E_{X},  \label{lpst_7}\\
                      z_e\geq 0, && \forall e\in E_{Z},  \label{lpst_8}\\
                      y_e\geq 0, && \forall e\in E_{Y}.  \label{lpst_9}
\end{eqnarray}
It is easily seen that if we set $E_{X}= E_{Z}= E_{Y}=E$, $V_{X}=V_{Y}=V$, $L=|V|-1-k$, 
then the linear program $LP_{RRST}(E_{X},V_{X},E_{Y},V_{Y}, E_{Z},L)$ is a linear programming relaxation of~(\ref{mpst}).
Indeed, the binary variables $x_e, y_e, z_e\in \{0,1\}$ indicate then the spanning trees $X$ and $Y$ and their common part $X\cap Y$, respectively.  Moreover, the constraint~(\ref{lpst_3}) takes the form of equality, instead of the inequality,
since the variables~$z_e$, $e\in E_Z$, are not present in the objective function~(\ref{lpst_o}).
Problem $LP_{RRST}(E_{X},V_{X},E_{Y},V_{Y}, E_{Z},L)$ has exponentially many constraints. However, the constraints (\ref{lpst_0}), (\ref{lpst_1}) and  (\ref{lpst_5}),  (\ref{lpst_6})
are the spanning tree ones for graphs $G_X=(V_X,E_X)$ and $G_Y=(V_Y,E_Y)$, respectively. Fortunately,
there exits a polynomial time separation oracle over such constraints~\cite{TLMLAW95}.
 Clearly, separating over the remaining constraints, i.e. (\ref{lpst_2}),  (\ref{lpst_3}) and  (\ref{lpst_4})
 can be done in a polynomial time.  In consequence, an optimal vertex solution to the problem can be found in polynomial time.   
 It is also worth pointing out that, alternatively, one 
 can  rewrite  (\ref{lpst_o}), (\ref{lpst_0}), (\ref{lpst_1}), (\ref{lpst_7}),   (\ref{lpst_5}) and  (\ref{lpst_9})
 in  a ``compact'' form that has the polynomial number of variables and constraints (see~\cite{TLMLAW95}).

Let us focus now on a vertex solution  $(\pmb{x},\pmb{z},\pmb{y})\in \Rset^{|E_X|\times |E_Z|\times |E_Y|}_{\geq 0}$ of the linear programming problem
$LP_{RRST}(E_{X},V_{X},E_{Y},V_{Y}, E_{Z},L)$. 
If $E_Z=\emptyset$, then the only constraints being left in (\ref{lpst_o})-(\ref{lpst_9})
are the spanning tree constraints. Thus $\pmb{x}$ and  $\pmb{y}$ are 0-1 incidence
vectors of
the spanning trees $X$ and $Y$, respectively (see~\cite[Theorem~3.2]{TLMLAW95}).

We now turn to a more involved case, when $E_X\not=\emptyset$,
 $E_Y\not=\emptyset$ and $E_Z\not=\emptyset$. 
 We first  reduce the sets $E_X$, $E_Y$ and $E_Z$ by removing  all edges~$e$
 with $x_e=0$, or $y_e=0$, or $z_e=0$. Removing these edges does not
 change the feasibility and the cost of the vertex solution $(\pmb{x},\pmb{z},\pmb{y})$. Note that $V_X$ and $V_Y$ remain unaffected.
 From now on, we can assume that variables corresponding to all edges from 
 $E_X$, $E_Y$ and $E_Z$ are positive, i.e. $x_e>0$, $e\in E_X$, $y_e>0$, $e\in E_Y$ and
 $z_e>0$, $e\in E_Z$. Hence the constraints (\ref{lpst_7}), (\ref{lpst_8}) and (\ref{lpst_9}) are
 not taken into account, since they are not tight with respect to $(\pmb{x},\pmb{z},\pmb{y})$. 
 It is possible, after reducing $E_X$, $E_Y$, and $E_Z$,
to  characterize  $(\pmb{x},\pmb{z},\pmb{y})$  by  
 $|E_X|+ |E_Z|+|E_Y|$ constraints that are linearly independent and tight with respect to $(\pmb{x},\pmb{z},\pmb{y})$.
 
 Let $\mathcal{F}(\pmb{x})=\{U\subseteq V_X\,:\, \sum_{e\in E_X(U)}x_e =|U|-1\}$ and 
 $\mathcal{F}(\pmb{y})=\{U\subseteq V_Y\,:\, \sum_{e\in E_Y(U)}y_e =|U|-1\}$ stand for
 the sets of subsets of nodes that indicate the tight constraints (\ref{lpst_0}), (\ref{lpst_1}) and 
  (\ref{lpst_5}), (\ref{lpst_6}) for $\pmb{x}$ and $\pmb{y}$, respectively.
  Similarly we define the sets of edges that indicate the tight constraints (\ref{lpst_2}) and (\ref{lpst_4}) 
  with respect to  $(\pmb{x},\pmb{z},\pmb{y})$, namely
     $\mathcal{E}(\pmb{x},\pmb{z})=\{e\in E_{X}\cap E_{Z}\,:\, -x_e+z_e =  0\}$ and
    $\mathcal{E}(\pmb{z},\pmb{y})=\{e\in E_{Y}\cap E_{Z}\,:\, z_e -y_e =  0\}$.
    Let $\chi_X(W)$, $W\subseteq E_X$, (resp. $\chi_Z(W)$, $W\subseteq E_Z$, and $\chi_Y(W)$, $W\subseteq E_Y$)
 denote the characteristic vector in $\{0,1\}^{|E_X|}\times \{0\}^{|E_Z|}\times\{0\}^{|E_Y|}$
(resp.  $\{0\}^{|E_X|}\times \{0,1\}^{|E_Z|}\times\{0\}^{|E_Y|}$ and $\{0\}^{|E_X|}\times \{0\}^{|E_Z|}\times\{0,1\}^{|E_Y|}$)
that has~1 if $e\in W$ and 0 otherwise.

 We recall that two sets $A$ and $B$ are \emph{intersecting} if
 $A\cap B$, $A\setminus B$,  $B\setminus A$ are nonempty.
 A family of sets is \emph{laminar} if no two sets are intersecting (see, e.g.,~\cite{LRS11}). Observe that the number of subsets in  $\mathcal{F}(\pmb{x})$ and  $\mathcal{F}(\pmb{y})$ can be exponential.
Let  $\mathcal{L}(\pmb{x})$ (resp.  $\mathcal{L}(\pmb{y})$) be
 a \emph{maximal laminar subfamily} of $\mathcal{F}(\pmb{x})$ (resp. $\mathcal{F}(\pmb{y})$).
 The following lemma, which is a slight extension of \cite[Lemma~4.1.5]{LRS11},
 allows us to choose out of  $\mathcal{F}(\pmb{x})$ and  $\mathcal{F}(\pmb{y})$  
 certain   subsets that indicate linearly independent tight constraints. 
 \begin{lem}
 For $\mathcal{L}(\pmb{x})$ and  $\mathcal{L}(\pmb{y})$ the  following equalities:
%$\mathrm{span}(\{ \chi_X(E_X(U))\,:\, U\in \mathcal{L}(\pmb{x})\})=
% \mathrm{span}(\{ \chi_X(E_X(U))\,:\, U\in \mathcal{F}(\pmb{x})\})$ and
% $\mathrm{span}(\{ \chi_Y(E_Y(U))\,:\, U\in \mathcal{L}(\pmb{y})\})=
% \mathrm{span}(\{ \chi_Y(E_Y(U))\,:\, U\in \mathcal{F}(\pmb{y})\})$ hold.
\begin{align*}
\mathrm{span}(\{ \chi_X(E_X(U))\,:\, U\in \mathcal{L}(\pmb{x})\})&=
 \mathrm{span}(\{ \chi_X(E_X(U))\,:\, U\in \mathcal{F}(\pmb{x})\}),\\
 \mathrm{span}(\{ \chi_Y(E_Y(U))\,:\, U\in \mathcal{L}(\pmb{y})\})&=
 \mathrm{span}(\{ \chi_Y(E_Y(U))\,:\, U\in \mathcal{F}(\pmb{y})\})
\end{align*}
hold.
 \label{lamlem}
 \end{lem}
 \begin{proof}
 The proof is the same as that for the  spanning tree in~\cite[Lemma~4.1.5]{LRS11}.
 \end{proof}
 A trivial verification shows that the following observation is true: 
\begin{obs}
$V_X\in \mathcal{L}(\pmb{x})$ and $V_Y\in \mathcal{L}(\pmb{y})$.
\label{lamob} 
\end{obs}
 
 We are now ready to give a characterization  of a vertex solution.
\begin{lem}
Let $(\pmb{x},\pmb{z},\pmb{y})$ be a vertex solution of 
$LP_{RRST}(E_{X},V_{X},E_{Y},V_{Y}, E_{Z})$  such that 
$x_e>0$, $e\in E_X$, $y_e>0$, $e\in E_Y$ and
 $z_e>0$, $e\in E_Z$. Then there exist laminar families $\mathcal{L}(\pmb{x})\not=\emptyset$
 and $\mathcal{L}(\pmb{y})\not=\emptyset$ and subsets 
 $E(\pmb{x},\pmb{z}) \subseteq \mathcal{E}(\pmb{x},\pmb{z})$ and
 $E(\pmb{z},\pmb{y}) \subseteq \mathcal{E}(\pmb{z},\pmb{y})$
 that must satisfy the following:
 \begin{itemize}
 \item[(i)] $|E_X|+|E_Z|+|E_Y|=|\mathcal{L}(\pmb{x})|+|E(\pmb{x},\pmb{z})|+|E(\pmb{z},\pmb{y})|
 +|\mathcal{L}(\pmb{y})|+1$,
 \item[(ii)]
 the vectors in
 $\{ \chi_X(E_X(U))\,:\, U\in \mathcal{L}(\pmb{x})\}\cup \{ \chi_Y(E_Y(U))\,:\, U\in \mathcal{L}(\pmb{y})\}
 \cup \{ -\chi_X(\{e\})+ \chi_Z(\{e\})  \,:\, e\in E(\pmb{x},\pmb{z})\}
 \cup \{ \chi_Z(\{e\})- \chi_Y(\{e\})  \,:\, e\in E(\pmb{z},\pmb{y})\} \cup \{\chi_Z(E_Z)\}$ 
 are linearly independent.
\end{itemize} 
 \label{ranklem}
\end{lem} 
\begin{proof}
The vertex $(\pmb{x},\pmb{z},\pmb{y})$ can be uniquely characterized by any set of
 linearly independent constraints
 with the cardinality of $|E_X|+ |E_Z|+|E_Y|$,
 chosen  from among the constraints
 (\ref{lpst_0})-(\ref{lpst_6}),  tight  with respect  to $(\pmb{x},\pmb{z},\pmb{y})$.
  We construct such set by choosing  a maximal subset of 
  linearly independent tight constraints that characterizes $(\pmb{x},\pmb{z},\pmb{y})$.
   Lemma~\ref{lamlem}  shows that there exist maximal laminar subfamilies 
   $\mathcal{L}(\pmb{x})\subseteq \mathcal{F}(\pmb{x})$ and
   $\mathcal{L}(\pmb{y})\subseteq \mathcal{F}(\pmb{y})$ 
   such that 
   $\mathrm{span}(\{ \chi_X(E_X(U))\,:\, U\in \mathcal{L}(\pmb{x})\})=
 \mathrm{span}(\{ \chi_X(E_X(U))\,:\, U\in \mathcal{F}(\pmb{x})\})$ and 
 $\mathrm{span}(\{ \chi_Y(E_Y(U))\,:\, U\in \mathcal{L}(\pmb{y})\})=
 \mathrm{span}(\{ \chi_Y(E_Y(U))\,:\, U\in \mathcal{F}(\pmb{y})\}))$.
 Observation~\ref{lamob} implies  $\mathcal{L}(\pmb{x})\not=\emptyset$ and 
  $\mathcal{L}(\pmb{y})\not=\emptyset$.
  Moreover, it is evident that
  $\mathrm{span}(\{ \chi_X(E_X(U))\,:\, U\in \mathcal{L}(\pmb{x})\} \cup 
  \{ \chi_Y(E_Y(U))\,:\, U\in \mathcal{L}(\pmb{y})\})=
  \mathrm{span}(\{ \chi_X(E_X(U))\,:\, U\in \mathcal{F}(\pmb{x})\} \cup 
  \{ \chi_Y(E_Y(U))\,:\, U\in \mathcal{F}(\pmb{y})\})$. Thus $\mathcal{L}(\pmb{x})\cup \mathcal{L}(\pmb{y})$
  indicate certain linearly independent tight constraints that have been already  included in the set constructed.
  We add (\ref{lpst_3})  to the set constructed. Obviously, it still consists of linearly independent constraints. 
  We complete forming the set by choosing  a maximal number of tight constraints 
  from among the ones (\ref{lpst_2}) and  (\ref{lpst_4}), such that they form 
  a linearly independent set with the constraints previously selected.
  We characterize these constraints by the sets of edges
  $E(\pmb{x},\pmb{z}) \subseteq \mathcal{E}(\pmb{x},\pmb{z})$ and
 $E(\pmb{z},\pmb{y}) \subseteq \mathcal{E}(\pmb{z},\pmb{y})$.
 Therefore 
 the vectors in
 $\{ \chi_X(E_X(U))\,:\, U\in \mathcal{L}(\pmb{x})\}\cup \{ \chi_Y(E_Y(U))\,:\, U\in \mathcal{L}(\pmb{y})\}
 \cup \{ -\chi_X(\{e\})+ \chi_Z(\{e\})  \,:\, e\in E(\pmb{x},\pmb{z})\}
 \cup \{ \chi_Z(\{e\})- \chi_Y(\{e\})  \,:\, e\in E(\pmb{z},\pmb{y})\} \cup \{\chi_Z(E_Z)\}$ 
 are linearly independent and represent the constructed  maximal set of  independent tight  constraints, 
 with the cardinality of  $|\mathcal{L}(\pmb{x})|+|E(\pmb{x},\pmb{z})|+|E(\pmb{z},\pmb{y})|
 +|\mathcal{L}(\pmb{y})|+1$, 
 that uniquely describe  $(\pmb{x},\pmb{z},\pmb{y})$.
 Hence  $|E_X|+|E_Z|+|E_Y|=|\mathcal{L}(\pmb{x})|+|E(\pmb{x},\pmb{z})|+|E(\pmb{z},\pmb{y})|
 +|\mathcal{L}(\pmb{y})|+1$ which establishes the lemma.
 \end{proof}
 \begin{lem} 
Let $(\pmb{x},\pmb{z},\pmb{y})$ be a vertex solution of 
$LP_{RRST}(E_{X},V_{X},E_{Y},V_{Y}, E_{Z})$  such that 
$x_e>0$, $e\in E_X$, $y_e>0$, $e\in E_Y$ and
 $z_e>0$, $e\in E_Z$. Then there is an edge $e'\in E_X$ with $x_{e'}=1$ or
 an edge $e''\in E_Y$ with $y_{e''}=1$.
 \label{lemx1y1}
\end{lem}
\begin{proof}
On  the contrary, suppose that $0<x_e<1$ for every $e\in E_X$  and $0<y_e<1$ for every $e\in E_Y$.
Constraints~(\ref{lpst_2}) and  (\ref{lpst_4})  lead to $0<z_e<1$ for every $e\in E_Z$.
By Lemma~\ref{ranklem} there exist laminar families $\mathcal{L}(\pmb{x})\not=\emptyset$
 and $\mathcal{L}(\pmb{y})\not=\emptyset$ and subsets 
 $E(\pmb{x},\pmb{z}) \subseteq \mathcal{E}(\pmb{x},\pmb{z})$ and
 $E(\pmb{z},\pmb{y}) \subseteq \mathcal{E}(\pmb{z},\pmb{y})$ indicating  linearly independent 
 constraints which uniquely define  $(\pmb{x},\pmb{z},\pmb{y})$, namely
 \begin{eqnarray}
                        \sum_{e\in E_X(U)} x_e =  |U|-1, && \forall U\in \mathcal{L}(\pmb{x}), \label{dlpst_1}\\
                         -x_e+z_e = 0,                              &&\forall e\in E(\pmb{x},\pmb{z}),  \label{dlpst_2}\\
                      \sum_{e\in E_Z} z_e= L, &&  \label{dlpst_3}\\
                       z_e -y_e = 0,                              &&\forall e\in E(\pmb{z},\pmb{y}),  \label{dlpst_4}\\
                      \sum_{e\in E_Y(U)} y_e =  |U|-1, && \forall U \in \mathcal{L}(\pmb{y}).  \label{dlpst_6}
\end{eqnarray}
We will arrive to a contradiction with Lemma~\ref{ranklem}(i)
 by applying a token counting argument frequently in used in~\cite{LRS11}.
  
We give exactly two tokens to each edge in  $E_X$, $E_Z$ and $E_Y$.
Thus we use $2|E_X|+2|E_Z|+2|E_Y|$ tokens.
 We then redistribute these tokens to the tight constraints (\ref{dlpst_1})-(\ref{dlpst_6}) as follows.
For $e\in E_X$ the first token is assigned to the constraint
indicated by the smallest set $U\in \mathcal{L}(\pmb{x})$ containing its two endpoints  (see~(\ref{dlpst_1}))
and the second one is assigned to the constraint represented by~$e$ (see~(\ref{dlpst_2}))
if  $e\in E(\pmb{x},\pmb{z})$.
Similarly,
for $e\in E_Y$ the first token is assigned to the constraint
indicated by the smallest set $U\in \mathcal{L}(\pmb{y})$ containing its both endpoints  (see~(\ref{dlpst_6}))
and the second one is assigned to the constraint represented by~$e$ (see~(\ref{dlpst_4})) if
$e\in E(\pmb{z},\pmb{y})$.
Each $e\in E_Z$ assigns the first token to the constraint corresponding to~$e$ (see~(\ref{dlpst_2}))
if  $e\in E(\pmb{x},\pmb{z})$; otherwise to  the constraint (\ref{dlpst_3}).
The second token is assigned to the constraint indicated by~$e$ (see~(\ref{dlpst_4}))
if  $e\in E(\pmb{z},\pmb{y})$.
\begin{cla}
Each of the constraints (\ref{dlpst_2}) and  (\ref{dlpst_4}) receives exactly two tokens.
Each of the constraints (\ref{dlpst_1}) and  (\ref{dlpst_6}) collects at least two tokens.
\label{dcla1}
\end{cla}
The first part of Claim~\ref{dcla1} is obvious. In order to show 
the second part
we apply the same reasoning 
as~\cite[Proof~2 of Lemma~4.2.1]{LRS11}.
Consider the constraint represented by any subset $U\in \mathcal{L}(\pmb{x})$.
We say that $U$ is the \emph{parent} of a subset $C\in  \mathcal{L}(\pmb{x})$
and $C$ is the \emph{child} of~$U$ if $U$ is  the smallest set containing~$C$.
Let $C_1,\ldots,C_{\ell}$ be the children of~$U$.
The constraints corresponding to these subsets are as follows
\begin{align}
 \sum_{e\in E_X(U)} x_e& =  |U|-1, \label{pc}\\
  \sum_{e\in E_X(C_k)} x_e &=  |C_k|-1, \; \forall k\in[\ell]. \label{cc}
\end{align} 
Subtracting (\ref{cc}) for every $k\in [\ell]$  from~(\ref{pc}) 
yields: 
\begin{equation}
\sum_{e\in  E_X(U)\setminus \bigcup_{k\in [\ell]}E_X(C_k)} x_e=|U|-\sum_{k\in [\ell]}|C_k| +\ell-1.
\label{ttc}
\end{equation}
Observe that $E_X(U)\setminus \bigcup_{k\in [\ell]}E_X(C_k)\not=\emptyset$.
Otherwise, this leads to a contradiction with the linear independence of the constraints.
Since the right hand side of (\ref{ttc}) is integer and $0<x_e<1$ for every $e\in E_X$,
$|E_X(U)\setminus \bigcup_{k\in [\ell]}E_X(C_k)|\geq 2$.
Hence $U$ receives at least two tokens.
The same arguments apply to the constraint represented by any subset $U\in \mathcal{L}(\pmb{y})$.
This proves the claim.
\begin{cla}
Either constraint~(\ref{dlpst_3}) collects at least one token and there are at least two extra tokens
left or
constraint~(\ref{dlpst_3})  receives no token and 
there are at least three extra tokens
left.
\label{dcla2}
\end{cla}
To prove the claim we need to consider several nested cases:
\begin{enumerate}
\item Case: $E_Z\setminus E(\pmb{x},\pmb{z})\not= \emptyset$.
Since $E_Z\setminus E(\pmb{x},\pmb{z})\not= \emptyset$,
at least one token is assigned to  constraint~(\ref{dlpst_3}).
We have yet to show that there are at least two token left.
\begin{enumerate}
\item Case: $E_Z\setminus E(\pmb{z},\pmb{y})=\emptyset$.
Subtracting (\ref{dlpst_4}) for every $e\in E(\pmb{z},\pmb{y})$ from~(\ref{dlpst_3})
gives:
\begin{equation}
\sum_{e\in E(\pmb{z},\pmb{y})} y_e= L.
\label{c2c1}
\end{equation}
\label{c1a}
\begin{enumerate}
\item Case: $E_Y\setminus E(\pmb{z},\pmb{y})=\emptyset$.
Thus $L=|V_Y|-1$, since $(\pmb{x},\pmb{z},\pmb{y})$ is a feasible solution.
By Observation~\ref{lamob}, $V_Y\in \mathcal{L}(\pmb{y})$ and (\ref{c2c1}) has the form of constraint~(\ref{dlpst_6})
for $V_Y$, which contradicts the linear independence of the constraints.
\item Case: $E_Y\setminus E(\pmb{z},\pmb{y})\not=\emptyset$.
Thus $L<|V_Y|-1$.
Since the right hand side of (\ref{c2c1}) is integer and $0<y_e<1$ for every $e\in E_Y$,
$|E_Y\setminus E(\pmb{z},\pmb{y})|\geq 2$.
Hence, there are  at least two extra  tokens left.
\end{enumerate}
\item Case: $E_Z\setminus E(\pmb{z},\pmb{y})\not=\emptyset$.
Consequently, 
$|E_Z\setminus E(\pmb{z},\pmb{y})|\geq 1$ and thus
at least one  token left over, i.e at least one token is not assigned to constraints~(\ref{dlpst_4}).
Therefore,
yet one additional token is required.
\label{c1b}
\begin{enumerate}
\item Case: $E_Y\setminus E(\pmb{z},\pmb{y})=\emptyset$.
Consider the constraint~(\ref{dlpst_6}) corresponding to~$V_Y$.
Adding (\ref{dlpst_4}) for every $e\in E(\pmb{z},\pmb{y})$ to this constraint
yields:
\begin{equation}
\sum_{e\in E(\pmb{z},\pmb{y})} z_e= |V_Y|-1.
\label{c2c2}
\end{equation}
Obviously $|V_Y|-1<L$.
Since $L$ is integer and $0<z_e<1$ for every $e\in E_Z$,
$|E_Z\setminus E(\pmb{z},\pmb{y})|\geq 2$.
Hence there are  at least two extra  tokens left.
\item Case: $E_Y\setminus E(\pmb{z},\pmb{y})\not=\emptyset$.
One sees immediately that 
 at least one  token left over, i.e at least one token  is not assigned to constraints~(\ref{dlpst_4}).
\end{enumerate}
\end{enumerate}
Summarizing the above cases, constraint~(\ref{dlpst_3}) collects at least one token and there are at least two extra tokens
left.
\item Case: $E_Z\setminus E(\pmb{x},\pmb{z})= \emptyset$.
Subtracting (\ref{dlpst_2}) for every $e\in E(\pmb{x},\pmb{y})$ from~(\ref{dlpst_3})
gives:
\begin{equation}
\sum_{e\in E(\pmb{x},\pmb{z})} x_e= L.
\label{c2c3}
\end{equation}
Thus constraint~(\ref{dlpst_3}) receives no token.
We need yet to show that there are at least three extra tokens left.
\begin{enumerate}
\item  Case: $E_X\setminus E(\pmb{x},\pmb{z})= \emptyset$.
Therefore $L=|V_X|-1$, since $(\pmb{x},\pmb{z},\pmb{y})$ is a feasible solution.
By Observation~\ref{lamob}, $V_X\in \mathcal{L}(\pmb{x})$ and (\ref{c2c3}) has the form of constraint~(\ref{dlpst_1})
for $V_X$, which contradicts with the linear independence of the constraints.
\item Case: $E_X\setminus E(\pmb{x},\pmb{z})\not= \emptyset$
Thus $L<|V_X|-1$.
Since the right hand side of (\ref{c2c3}) is integer and $0<x_e<1$ for every $e\in E_X$,
$|E_X\setminus E(\pmb{x},\pmb{z})|\geq 2$. Consequently,
there are  at least two extra  tokens left.
Yet at least one token is required.
\begin{enumerate}
\item Case: $E_Z\setminus E(\pmb{z},\pmb{y})= \emptyset$. 
Reasoning  is the same as in Case~\ref{c1a}.
\item Case: $E_Z\setminus E(\pmb{z},\pmb{y})\not= \emptyset$.
Reasoning  is the same as in  Case~\ref{c1b}.
\end{enumerate}
\end{enumerate}
Accordingly, constraint~(\ref{dlpst_3})  receives no token and 
there are at least three extra tokens
left.
\end{enumerate}
Thus the claim is proved.
The method of assigning tokens to constraints   (\ref{dlpst_1})-(\ref{dlpst_6}) and
Claims~\ref{dcla1} and \ref{dcla2}    now show that
either
\[
2|E_X|+2|E_Z|+2|E_Y|-2\geq2|\mathcal{L}(\pmb{x})|+2|E(\pmb{x},\pmb{z})|+2|E(\pmb{z},\pmb{y})|
+2|\mathcal{L}(\pmb{y})|
+1
\]
or
\[
2|E_X|+2|E_Z|+2|E_Y|-3\geq2|\mathcal{L}(\pmb{x})|+2|E(\pmb{x},\pmb{z})|+2|E(\pmb{z},\pmb{y})| 
+2|\mathcal{L}(\pmb{y})|.
\]
The above inequalities lead to
$|E_X|+|E_Z|+|E_Y|>|\mathcal{L}(\pmb{x})|+|E(\pmb{x},\pmb{z})|+|E(\pmb{z},\pmb{y})|
 +|\mathcal{L}(\pmb{y})|+1$. This contradicts Lemma~\ref{ranklem}(i).
\end{proof}

It remains to verify two cases: $E_X=\emptyset$ and $|V_X|=1$;  $E_Y=\emptyset$ and $|V_Y|=1$.
We consider only for the first one, the second case is symmetrical. 
Then
  the constraints (\ref{lpst_3}), (\ref{lpst_4}) and the inclusion 
$E_{Z}\subseteq E_{Y}$ yield 
\begin{equation}
\sum_{e\in E_Z} y_e\geq L. \label{um}
\end{equation}
 \begin{lem}
Let $\pmb{y}$ be a vertex solution of  linear program: 
(\ref{lpst_o}),  (\ref{lpst_5}),  (\ref{lpst_6}),  (\ref{lpst_9}) and  (\ref{um}) such that
 $y_e>0$, $e\in E_Y$. Then 
there is an edge $e'\in E_Y$ with $y_{e'}=1$.
Moreover, 
using 
$\pmb{y}$ one can  construct a vertex solution of
 $LP_{RRST}(\emptyset,V_{X},E_{Y},V_{Y}, E_{Z})$ with $y_{e'}=1$ and  the cost of~$\pmb{y}$.
 \label{lemy1}
\end{lem}
\begin{proof}
 Similarly as in the proof Lemma~\ref{ranklem} 
 we construct  a maximal subset of 
  linearly independent tight constraints that characterize~ $\pmb{y}$ and
  get:  $|E_Y|=|\mathcal{L}(\pmb{y})|$ if  (\ref{um}) is not tight or  adding (\ref{um}) makes the subset dependent;
         $|E_Y|=|\mathcal{L}(\pmb{y})|+1$ otherwise. In the first case the spanning tree constraints define~$\pmb{y}$
         and, in consequence, $\pmb{y}$ is integral (see~\cite[Theorem~3.2]{TLMLAW95}).
 Consider the second case and assume, on the contrary,
    that $0<y_e<1$ for each $e\in E_Y$. Thus
  \begin{eqnarray}
                      \sum_{e\in E_Z} y_e= L, &&  \label{dlpst_3l2}\\
                      \sum_{e\in E_Y(U)} y_e =  |U|-1, && \forall U \in \mathcal{L}(\pmb{y}).  \label{dlpst_6l2}
\end{eqnarray}         
We assign two tokens   to each edge in  $E_Y$ and  redistribute $2|E_Y|$ tokens to constraints (\ref{dlpst_3l2})
and (\ref{dlpst_6l2}) in the following way. 
The first token is given to the constraint
indicated by the smallest set $U\in \mathcal{L}(\pmb{y})$ containing its two endpoints and the second
one is  assigned to  (\ref{dlpst_3l2}). Since $0<y_e<1$ and $L$ is integer,
 similarly as in the proof Lemma~\ref{lemx1y1}, one can show that
 each of the constraints~(\ref{dlpst_6l2}) and (\ref{dlpst_3l2}) receives at least two tokens.
 If $E_Y\setminus E_Z= \emptyset$ then $L=|V_Y|-1$ since $\pmb{y}$ is a feasible solution - a contradiction
 with 
 the linear independence of the constraints. Otherwise ($E_Y\setminus E_Z\not= \emptyset$)  at least one token is left.
 Hence $2|E_Y|-1=2|\mathcal{L}(\pmb{y})|+2$ and so $|E_Y|>|\mathcal{L}(\pmb{y})|+1$, a contradiction.
 
 By (\ref{um})  and the fact that there are no variables $z_e$, $e\in E_Z$,  in the objective~(\ref{lpst_o}),
  it is obvious   that using~$\pmb{y}$ one can construct~$\pmb{z}$ satisfying~(\ref{lpst_3}) 
  and, in consequence, 
  a vertex solution of
 $LP_{RRST}(\emptyset,V_{X},E_{Y},V_{Y}, E_{Z})$ with $y_{e'}=1$ and  the cost of~$\pmb{y}$.
 \end{proof}
\begin{algorithm}
\begin{small}
$E_X\leftarrow E$, $E_Y\leftarrow E$, $E_Z\leftarrow E$,  $V_X\leftarrow V$,
$V_Y\leftarrow V$, $L\leftarrow |V|-1-k$, $X\leftarrow \emptyset$, $Y\leftarrow \emptyset$, $Z\leftarrow \emptyset$\;
\While{$|V_X|\geq 2$ \KwOR $|V_Y|\geq 2$}{
Find an optimal vertex solution $(\pmb{x}^*,\pmb{z}^*,\pmb{y}^*)$  of  $LP_{RRST}(E_{X},V_{X},E_{Y},V_{Y}, E_{Z},L)$\;
\label{steplp}
\ForEach{$e\in E_Z$ with $z^{*}_e=0$}{$E_Z\leftarrow E_Z\setminus \{e\}$\;}
\label{stepz0}
\ForEach{$e\in E_X$ with $x^{*}_e=0$}{$E_X\leftarrow E_X\setminus \{e\}$\;}
\label{stepx0}
\ForEach{$e\in E_Y$ with $y^{*}_e=0$}{$E_Y\leftarrow E_Y\setminus \{e\}$\;}
\label{stepy0}
\If{there exists edge $e'\in E_X$ with $x^{*}_{e'}=1$ \label{stepx1}}{
$X\leftarrow X\cup \{e'\}$  \; 
contract edge~$e'=\{u,v\}$ by deleting~$e$ and  identifying its endpoints~$u$ and $v$
 in graph $G_X=(V_X,E_X)$, induced by $V_X$ and $E_X$,
which is equivalent to: $|V_X|\leftarrow |V_X|-1$ and
$E_X\leftarrow E_X\setminus \{e'\}$\;  \label{stepxx1}
}
\If{there exists edge $e'\in E_X$ with $y^{*}_{e'}=1$\label{stepy1}}{
$Y\leftarrow Y\cup \{e'\}$\;
contract edge~$e'=\{u,v\}$ by deleting~$e$ and identifying its endpoints~$u$ and $v$
 in graph $G_Y=(V_Y,E_Y)$, induced by $V_Y$ and $E_Y$,
which is equivalent to: $|V_Y|\leftarrow |V_Y|-1$ and
$E_Y\leftarrow E_Y\setminus \{e'\}$\;\label{stepyy1}
}
\If{there exists edge $e'\in E_Z$ such that  $e'\in X\cap Y$ \label{stepzxy}}{
\tcc{Here $(\pmb{x}^*,\pmb{z}^*,\pmb{y}^*)$ with $\sum_{e\in E_Z}z^{*}_e=L$
can be always converted, preserving the cost, to 
$(\pmb{x}^*,\pmb{z}',\pmb{y}^*)$ with $z'_{e'}=1$ and $\sum_{e\in E_Z\setminus \{e'\}}z'_e=L-1$}
$Z\leftarrow Z\cup \{e'\}$, $L\leftarrow L-1$\;
$E_Z\leftarrow E_Z\setminus \{e'\}$\; \label{stepzxyz}
}
}
\Return{$X$, $Y$, $Z$}
  \caption{Algorithm for  RR ST}
 \label{algrrst}
\end{small} 
\end{algorithm}
 We are now ready to give the main result of this section.
\begin{thm}
Algorithm~\ref{algrrst} solves  RR~ST in polynomial time.
\end{thm} 
\begin{proof}
Lemmas~\ref{lemx1y1} and~\ref{lemy1} and the case when $E_Z=\emptyset$ (see the comments in this section)
ensure that Algorithm~\ref{algrrst} terminates after performing $O(|V|)$ iterations (Steps~\ref{steplp}-\ref{stepzxyz}).
Let  $\mathrm{OPT_{LP}}$  denote the optimal objective function value of 
$LP_{RRST}(E_X,V_X,E_Y,V_Y, E_Z,L)$,
where $E_{X}= E_{Z}= E_{Y}=E$, $V_{X}=V_{Y}=V$, $L=|V|-1-k$. Hence $\mathrm{OPT_{LP}}$ is a lower bound on the optimal objective value of RR~ST. 
It is not difficult to show that after the termination of the algorithm $X$ and $Y$ are two spanning trees in $G$ such that $\sum_{e\in X} C_e + \sum_{e\in Y} \overline{c}_e\leq \mathrm{OPT_{LP}}$. It remains to show that $|X\cap Y|\geq |V|-k-1$. By induction on the number of iterations of Algorithm~\ref{algrrst}
one can easily show that at any iteration
the inequality $L+|Z|= |V|-1-k$ is satisfied.
Accordingly,
if, after the termination of the algorithm, it holds $L=0$, then we are done. Suppose, on the  contrary that $L\geq 1$ ($L$ is integer). Since $\sum_{e\in E_Z}z^{*}_e\geq L$, $|E_Z|\geq L\geq 1$ and
$E_Z$ is the set with edges not belonging to $ X\cap Y$.
Consider any $e'\in E_Z$. Of course $z^{*}_{e'}>0$ and it remained positive during
the course of Algorithm~\ref{algrrst}.
Moreover, at least one of the constraints $-x_{e'} +z_{e'}\leq 0$ or  $z_{e'} -y_{e'}\leq 0$ is still present in the linear program (\ref{lpst_o})-(\ref{lpst_9}).
Otherwise, since $z^{*}_{e'}>0$, Steps~\ref{stepx1}-\ref{stepxx1},  Steps~\ref{stepy1}-\ref{stepyy1} and,
in consequence, Steps~\ref{stepzxy}-\ref{stepzxyz} for~$e'$ must have been executed 
during
the course of Algorithm~\ref{algrrst} and 
$e'$ has been included to~$Z$, a contradiction with the fact $e'\not\in X\cap Y$.
Since
the above constraints are present, $0<x^{*}_{e'}<1$ or $0<y^{*}_{e'}<1$. Thus
$e'\in E_X$ or  $e'\in E_Y$, which contradicts the termination of Algorithm~\ref{algrrst}.
\end{proof}
 
\section{Robust recoverable   matroid basis problem}
\label{srrmb}
The minimum spanning tree
can be generalized to the following  \emph{minimum  matroid basis  problem}.
We are given a matroid $M=(E,\mathcal{I})$ (see~\cite{o92}), where  $E$ is a nonempty ground set, $|E|=m$,
 and 
 $\mathcal{I}$ is a family of subsets of $E$, called \emph{independent sets}. The following two axioms must be satisfied: (i) if $A\subseteq B$ and $B\in \mathcal{I}$,
then $A\in \mathcal{I}$; (ii) for all $A,B\in \mathcal{I}$ if $|A|< |B|$, then there is an element $e \in B \setminus A$
such that $A \cup \{e\} \in \mathcal{I}$. 
We make the assumption that checking the independence of a set $A \subseteq E$ can be done in polynomial time.  
The \emph{rank function} of~$M$, $r_M: 2^E\rightarrow \Zset_{\geq 0}$, 
is defined by $r_M(U)=\max_{W \subseteq U, W \in \mathcal{I}}|W|$.
A \emph{basis} of $M$ is a maximal  under inclusion element of $\mathcal{I}$. The cardinality of each basis equals $r_M(E)$.
Let $c_e$ be a cost specified for each element $e\in E$. We wish to find a basis of~$M$  of the minimum total cost, $\sum_{e\in M}c_e$. It is well known that the minimum matroid basis problem is polynomially solvable by a greedy algorithm (see.~\cite{E71}). A spanning tree is a basis of a \emph{graphic matroid}, in which $E$ is a set of edges of a given graph and $\mathcal{I}$ is the set of all forests in $G$.

We now define two operations on matroid $M=(E,\mathcal{I})$, used in the following (see also~\cite{o92}).
Let $M\backslash e=(E_{M\backslash  e},\mathcal{I}_{M \backslash e})$, the \emph{deletion~$e$ from~$M$},
be the matroid obtained by deleting~$e\in E$  from~$M$ defined by
$E_{M\backslash e}=E\setminus\{e\}$ and 
$\mathcal{I}_{M\backslash e}=\{U\subseteq E\setminus\{e\}\,:\, U\in \mathcal{I}\}$.
The rank function of $M\backslash e$ is given by 
$r_{M\backslash e}(U)=r_{M}(U)$ for all $U\subseteq E\setminus\{e\}$.
Let  $M/e=(E_{M/e},\mathcal{I}_{M/e})$, the \emph{contraction~$e$ from~$M$},
be the matroid obtained by contracting~$e \in E$  in~$M$, defined by
$E_{M/e}=E\setminus\{e\}$; and 
$\mathcal{I}_{M/e}=\{U\subseteq E\setminus\{e\}\,:\, U\cup \{e\}\in \mathcal{I}\}$
if $\{e\}$ is independent and  $\mathcal{I}_{M/e}=\mathcal{I}$, otherwise.
The rank function of $M/e$ is given by 
$r_{M/e}(U)=r_{M}(U)-r_{M}(\{e\})$ for all $U\subseteq E\setminus\{e\}$.

Assume now that the first stage cost of element $e\in E$ equals $C_e$ and
its second stage cost
 is  uncertain and is modeled by interval $[c_e, c_e+d_e]$.
The \emph{robust recoverable   matroid basis problem} (RR MB for short) under 
scenario set $\mathcal{U}^m$
can be stated similarly to RR ST. Indeed, it suffices
to replace  the set of spanning trees by the bases of~$M$ and $\mathcal{U}$
by $\mathcal{U}^m$
in the formulation~(\ref{rrp}) and, in consequence, in~(\ref{rrst}). 
Here and subsequently, $\Phi$ denotes the set of all bases of~$M$. 
Likewise, RR MB under $\mathcal{U}^m$
is equivalent to the following problem:
\begin{equation}
 \begin{array}{ll}
 \min &\sum_{e\in X}C_e + \sum_{e\in Y}(c_e+d_e)\\
\text{s.t. } & |X \cap Y| \geq r_{M}(E)-k, \\
       & X,Y\in \Phi.
 \end{array}
 \label{mpmb}
\end{equation}

Let $E_{X}, E_{Y}\subseteq E$ and $\mathcal{I}_{X}$, $\mathcal{I}_{Y}$ be collections of subsets of $E_{X}$ and
$E_{Y}$, respectively, (independent sets),
that induce matroids ${M}_X=(E_X, \mathcal{I}_X)$ and ${M}_Y=(E_Y, \mathcal{I}_Y)$.
 Let $E_Z$ be a subset of $E$ such that $E_Z \subseteq E_X \cup E_Z$ and $|E_Z|\geq L$ for some fixed $L$. 
 The following linear program, denoted by $LP_{RRMB}(E_{X},\mathcal{I}_{X},E_{Y},\mathcal{I}_{Y},E_{Z},L)$,
 after setting $E_X=E_Y=E_Z=E$, $\mathcal{I}_X=\mathcal{I}_Y=\mathcal{I}$ and
 $L=r_{M}(E)-k$
 is a relaxation of~(\ref{mpmb}):
\begin{eqnarray}
    \min       \sum_{e\in E_X}C_e x_e + \sum_{e\in E_{Y}}(c_e+d_e)y_e&& \label{lpmb_o}\\
    \text{s.t. } \sum_{e\in E_X} x_e = r_{M_{X}}(E_X), && \label{lpmb_0} \\
                      \sum_{e\in U} x_e \leq  r_{M_{X}}(U), && \forall U \subset E_X, \label{lpmb_1}\\
                      -x_e+z_e \leq  0,                              &&\forall e\in E_{X}\cap E_{Z},  \label{lpmb_2}\\
                      \sum_{e\in E_Z} z_e= L, &&  \label{lpmb_3}\\
                       z_e -y_e \leq  0,                              &&\forall e\in E_{Y}\cap E_{Z},  \label{lpmb_4}\\
                       \sum_{e\in E_Y} y_e = r_{M_{Y}}(E_Y), &&   \label{lpmb_5}\\
                      \sum_{e\in U} y_e \leq  r_{M_{Y}}(U), && \forall U\subset E_Y,  \label{lpmb_6}\\
                      x_e\geq 0, && \forall e\in E_{X},  \label{lpmb_7}\\
                      z_e\geq 0, && \forall e\in E_{Z},  \label{lpmb_8}\\
                      y_e\geq 0, && \forall e\in E_{Y}  \label{lpmbb_9}.
\end{eqnarray}
The indicator variables $x_e, y_e, z_e\in \{0,1\}$, $e\in E$, 
 describe the bases $X$, $Y$ and their intersection $X\cap Y$, respectively have been relaxed.
Since there are no variables~$z_e$ in the objective~(\ref{lpmb_o}),
we can use equality constraint (\ref{lpmb_3}), instead of the inequality one.
The above linear program is solvable in polynomial time. 
The rank constraints (\ref{lpmb_0}), (\ref{lpmb_1}) and (\ref{lpmb_5}), (\ref{lpmb_6}) 
relate to matroids ${M}_X=(E_X, \mathcal{I}_X)$ and ${M}_Y=(E_Y, \mathcal{I}_Y)$, respectively,
and  a separation over these constraints can be carried out  in polynomial time~\cite{C84}.
Obviously, a separation over (\ref{lpmb_2}),  (\ref{lpmb_3}) and  (\ref{lpmb_4}) can be
done in polynomial time as well.

Consider 
a vertex solution $(\pmb{x},\pmb{z},\pmb{y})\in \Rset^{|E_X|\times |E_Z|\times |E_Y|}_{\geq 0}$ 
of $LP_{RRMB}(E_{X},\mathcal{I}_{X},E_{Y},\mathcal{I}_{Y}, E_{Z},L)$.
Note that if $E_Z=\emptyset$, then  the only  rank constraints 
are left 
in  (\ref{lpmb_o})-(\ref{lpmbb_9}).
Consequently,  $\pmb{x}$ and  $\pmb{y}$ are 0-1 incidence
vectors of
bases $X$ and $Y$ of matroids  ${M}_X$ and  ${M}_Y$, respectively (see~\cite{E71}).
Let us turn to other cases.
Assume that $E_X\not=\emptyset$,
 $E_Y\not=\emptyset$ and $E_Z\not=\emptyset$. 
Similarly as in Section~\ref{srrsp} we first reduce the sets $E_X$, $E_Y$ and $E_Z$ by removing  all elements~$e$
with $x_e=0$, or $y_e=0$, or $z_e=0$.
Let $\mathcal{F}(\pmb{x})=\{U\subseteq E_X\,:\, \sum_{e\in U}x_e =r_{M_X}(U)\}$ and 
$\mathcal{F}(\pmb{y})=\{U\subseteq E_Y\,:\, \sum_{e\in U}y_e =r_{M_Y}(U)\}$ denote
the sets of subsets of elements that indicate tight constraints (\ref{lpmb_0}), (\ref{lpmb_1}) and (\ref{lpmb_5}), (\ref{lpmb_6}) for $\pmb{x}$ and $\pmb{y}$, respectively.
Similarly we define the sets of elements that indicate tight constraints (\ref{lpmb_2}) and (\ref{lpmb_4}) with respect to  $(\pmb{x},\pmb{z},\pmb{y})$, namely
 $\mathcal{E}(\pmb{x},\pmb{z})=\{e\in E_{X}\cap E_{Z}\,:\, -x_e+z_e =  0\}$ and
 $\mathcal{E}(\pmb{z},\pmb{y})=\{e\in E_{Y}\cap E_{Z}\,:\, z_e -y_e =  0\}$.
 Let $\chi_X(W)$, $W\subseteq E_X$, (resp. $\chi_Z(W)$, $W\subseteq E_Z$, and $\chi_Y(W)$, $W\subseteq E_Y$)
 denote the characteristic vector in $\{0,1\}^{|E_X|}\times \{0\}^{|E_Z|}\times\{0\}^{|E_Y|}$
(resp.  $\{0\}^{|E_X|}\times \{0,1\}^{|E_Z|}\times\{0\}^{|E_Y|}$ and $\{0\}^{|E_X|}\times \{0\}^{|E_Z|}\times\{0,1\}^{|E_Y|}$)
that has~1 if $e\in W$ and 0 otherwise.

We recall that a family $\mathcal{L} \subseteq 2^E$ is a \emph{chain} if for any $A, B \in \mathcal{L}$, either $A \subseteq B$ or $B \subseteq A$
 (see, e.g.,~\cite{LRS11}).
 Let  $\mathcal{L}(\pmb{x})$ (resp.  $\mathcal{L}(\pmb{y})$) be
 a \emph{maximal chain subfamily} of $\mathcal{F}(\pmb{x})$ (resp. $\mathcal{F}(\pmb{y})$).
 The following lemma is a fairly straightforward
adaptation  of \cite[Lemma~5.2.3]{LRS11} to the problem under consideration and its proof
 may be handled in much the same way.
 \begin{lem}
 For $\mathcal{L}(\pmb{x})$ and  $\mathcal{L}(\pmb{y})$ the  following equalities:
%$ \mathrm{span}(\{ \chi_X(E_X(U))\,:\, U\in \mathcal{L}(\pmb{x})\})=
% \mathrm{span}(\{ \chi_X(E_X(U))\,:\, U\in \mathcal{F}(\pmb{x})\})$ and
% $\mathrm{span}(\{ \chi_Y(E_Y(U))\,:\, U\in \mathcal{L}(\pmb{y})\})=
% \mathrm{span}(\{ \chi_Y(E_Y(U))\,:\, U\in \mathcal{F}(\pmb{y})\})$ hold.
\begin{align*}
 \mathrm{span}(\{ \chi_X(E_X(U))\,:\, U\in \mathcal{L}(\pmb{x})\})&=
 \mathrm{span}(\{ \chi_X(E_X(U))\,:\, U\in \mathcal{F}(\pmb{x})\}),\\
 \mathrm{span}(\{ \chi_Y(E_Y(U))\,:\, U\in \mathcal{L}(\pmb{y})\})&=
 \mathrm{span}(\{ \chi_Y(E_Y(U))\,:\, U\in \mathcal{F}(\pmb{y})\})
\end{align*}
hold.
 \label{chainlem}
 \end{lem}
 
The next lemma, which characterizes a vertex solution,  is analogous to Lemma~\ref{ranklem}.
Its proof is based on Lemma~\ref{chainlem} and is similar in spirit to the one of Lemma~\ref{ranklem}.
\begin{lem}
Let $(\pmb{x},\pmb{z},\pmb{y})$ be a vertex solution of 
$LP_{RRMB}(E_{X},\mathcal{I}_{X},E_{Y},\mathcal{I}_{Y},E_{Z},L)$ such that 
$x_e>0$, $e\in E_X$, $y_e>0$, $e\in E_Y$ and
 $z_e>0$, $e\in E_Z$. Then there exist chain families $\mathcal{L}(\pmb{x})\not=\emptyset$
 and $\mathcal{L}(\pmb{y})\not=\emptyset$ and subsets 
 $E(\pmb{x},\pmb{z}) \subseteq \mathcal{E}(\pmb{x},\pmb{z})$ and
 $E(\pmb{z},\pmb{y}) \subseteq \mathcal{E}(\pmb{z},\pmb{y})$
 that must satisfy the following:
 \begin{itemize}
 \item[(i)] $|E_X|+|E_Z|+|E_Y|=|\mathcal{L}(\pmb{x})|+|E(\pmb{x},\pmb{z})|+|E(\pmb{z},\pmb{y})|
 +|\mathcal{L}(\pmb{y})|+1$,
 \item[(ii)]
 the vectors in
 $\{ \chi_X(E_X(U))\,:\, U\in \mathcal{L}(\pmb{x})\}\cup \{ \chi_Y(E_Y(U))\,:\, U\in \mathcal{L}(\pmb{y})\}
 \cup \{ -\chi_X(\{e\})+ \chi_Z(\{e\})  \,:\, e\in E(\pmb{x},\pmb{z})\}
 \cup \{ \chi_Z(\{e\})- \chi_Y(\{e\})  \,:\, e\in E(\pmb{z},\pmb{y})\} \cup \{\chi_Z(E_Z)\}$ 
 are linearly independent.
\end{itemize} 
 \label{rankchainlem}
\end{lem} 
Lemmas~\ref{chainlem} and \ref{rankchainlem}  now lead to the next two ones and their 
proofs  run as the proofs of Lemmas~\ref{lemx1y1} and~\ref{lemy1}.
 \begin{lem}
Let $(\pmb{x},\pmb{z},\pmb{y})$ be a vertex solution of 
$LP_{RRMB}(E_{X},\mathcal{I}_{X},E_{Y},\mathcal{I}_{Y},E_{Z},L)$   such that 
$x_e>0$, $e\in E_X$, $y_e>0$, $e\in E_Y$ and
 $z_e>0$, $e\in E_Z$. Then there is an element $e'\in E_X$ with $x_{e'}=1$ or
 an element $e''\in E_Y$ with $y_{e''}=1$.
 \label{lemx1y1chain}
\end{lem}
 We now turn to
two cases: $E_X=\emptyset$;  $E_Y=\emptyset$.
Consider $E_X=\emptyset$, the second case is symmetrical. 
Observe that (\ref{lpmb_3}) and (\ref{lpmb_4}) and 
$E_{Z}\subseteq E_{Y}$ implies constraint~(\ref{um}).
 \begin{lem}
Let $\pmb{y}$ be a vertex solution of  linear program: 
(\ref{lpmb_o}),  (\ref{lpmb_5}),  (\ref{lpmb_6}),  (\ref{lpmbb_9}) and  (\ref{um}) such that
 $y_e>0$, $e\in E_Y$. Then 
there is an element $e'\in E_Y$ with $y_{e'}=1$.
Moreover 
using 
$\pmb{y}$ one can  construct a vertex solution of
$LP_{RRMB}(\emptyset,\emptyset,E_{Y},\mathcal{I}_{Y},E_{Z},L)$ 
 with $y_{e'}=1$ and  the cost of~$\pmb{y}$.
 \label{lemmby1}
\end{lem}

\begin{algorithm}
\begin{small}
$M_X=(E_{X},\mathcal{I}_{X})\leftarrow (E,\mathcal{I})$,
$M_Y=(E_{Y},\mathcal{I}_{Y})\leftarrow (E,\mathcal{I})$, $L\leftarrow r_{M}(E)-k$,
$X\leftarrow \emptyset$, $Y\leftarrow \emptyset$, $Z\leftarrow \emptyset$\;
\While{$E_X\not= \emptyset$  \KwOR $E_Y\not= \emptyset$}{
Find an optimal vertex solution $(\pmb{x}^*,\pmb{z}^*,\pmb{y}^*)$  of  $LP_{RRMB}(E_{X},\mathcal{I}_{X},E_{Y},\mathcal{I}_{Y},E_{Z},L)$\;
\label{steplpmb}
\ForEach{$e\in E_Z$ with $z^{*}_e=0$ }{$E_Z\leftarrow E_Z\setminus \{e\}$\;}
\label{stepz0mb}
\ForEach{$e\in E_X$ with $x^{*}_e=0$ }{$M_{X}\leftarrow M_{X}\backslash e$\;}
\label{stepx0mb}
\ForEach{$e\in E_Y$ with $y^{*}_e=0$ }{$M_{Y}\leftarrow M_{Y}\backslash e$\;}
\label{stepy0mb}
\If{there exists element $e\in E_X$ with $x^{*}_e=1$ \label{stepx1mb}}{
$X\leftarrow X\cup \{e\}$  \; 
$M_{X}\leftarrow M_{X}/ e$ \;  \label{stepxx1mb}
}
\If{there exists element $e\in E_X$ with $y^{*}_e=1$\label{stepy1mb}}{
$Y\leftarrow Y\cup \{e\}$\;
$M_{Y}\leftarrow M_{Y}/ e$\; \label{stepyy1mb}
}
\If{there exists element $e\in E_Z$ such that  $e\in X\cap Y$ \label{stepzxymb}}{
\tcc{Here $(\pmb{x}^*,\pmb{z}^*,\pmb{y}^*)$ with $\sum_{e\in E_Z}z^{*}_e=L$
can be always converted, preserving the cost, to 
$(\pmb{x}^*,\pmb{z}',\pmb{y}^*)$ with $z'_{e'}=1$ and $\sum_{e\in E_Z\setminus \{e'\}}z'_e=L-1$}
$Z\leftarrow Z\cup \{e\}$,$L\leftarrow L-1$\;
$E_Z\leftarrow E_Z\setminus \{e\}$\; \label{stepzxyzmb}
}
}
\Return{$X$, $Y$, $Z$}

  \caption{Algorithm for  RR MB}
 \label{algrrmb}
\end{small} 
\end{algorithm}
We are thus led to
 the main result of this section. Its proof 
 follows by the same arguments as for RR ST.
\begin{thm}
Algorithm~\ref{algrrmb} solves RR~MB in polynomial time.
\end{thm}

\section{Conclusions}

In this paper we have shown that the recoverable version of the minimum spanning tree problem with interval edge costs is polynomially solvable. We have thus resolved a problem which has been open to date.  We have applied a technique called an iterative relaxation. 
It has been turned out that
the algorithm proposed for the minimum spanning tree can be easily generalized to the recoverable version of the matroid basis problem with interval element costs. Our polynomial time algorithm is based on 
solving linear programs. Thus
 the next step should be designing a polynomial time combinatorial algorithm for this problem, which is an interesting subject of further research.

%***********Acknowledgements***********************************
\subsubsection*{Acknowledgements}
The first author was supported by Wroc{\l}aw University of Technology,
grant S50129/K1102. 
The second and the third authors were
supported by
 the National Center for Science (Narodowe Centrum Nauki), grant  2013/09/B/ST6/01525.

%******************References*****************************************
%\bibliographystyle{abbrv}
%\bibliography{robust}

\end{document}